\documentclass[a4paper,10pt]{article}
\usepackage[margin=0.92in]{geometry}

\usepackage{microtype}

\usepackage{amsmath}
\usepackage{amssymb}
\usepackage{amsthm}
\theoremstyle{definition}
\newtheorem{definition}{Definition}[section]
\newtheorem{theorem}{Theorem}[section]
\newtheorem{proposition}[theorem]{Proposition}
\newtheorem{lemma}[theorem]{Lemma}

\usepackage[labelformat=simple]{subcaption}
 
\usepackage{enumerate} 
\usepackage{thm-restate}
\usepackage{alltt} 
\usepackage{color} 
\usepackage{graphicx} 
\usepackage{url}
\usepackage{latexsym}

\usepackage{hyperref}

\newcommand{\complexf}[1]{\mathcal{#1}}   
\newcommand{\cplC}{\complexf{C}}   
\newcommand{\cplD}{\complexf{D}}   
\newcommand{\cplI}{\complexf{I}}   
\newcommand{\cplO}{\complexf{O}}   

\newcommand{\keywd}[1]{\emph{#1}}       

\newcommand{\Pfont}[1]{\mathsf{#1}}
\newcommand{\Pinput}[2]{\Pfont{input}_{#1}^{#2}}
\newcommand{\Pdecide}[2]{\Pfont{decide}_{#1}^{#2}}

\newcommand{\Formulafont}[1]{\mathrm{#1}}
\newcommand{\DEC}[1]{\Formulafont{DEC}_{#1}}
\newcommand{\VALID}{\Formulafont{VALID}}
\newcommand{\KNOW}{\Formulafont{KNOW}}
\newcommand{\AGREE}[1]{\Formulafont{AGREE}_{#1}}
\newcommand{\OFUN}{\Formulafont{OFUN}}
\newcommand{\IFUN}{\Formulafont{IFUN}}
\newcommand{\Modalfont}[1]{\mathrm{#1}}
\newcommand{\ModK}[1]{\mathop{\Modalfont{K}_{#1}}} 
\newcommand{\ModD}[1]{\mathop{\Modalfont{D}_{#1}}} 
\newcommand{\ModC}[1]{\mathop{\Modalfont{C}_{#1}}} 
\newcommand{\AtomProps}{\mathrm{AP}}                 
\newcommand{\PropVar}{\mathrm{PV}}                 

\newcommand{\semFormula}[3]{\lVert{}#3\rVert^{#1}_{#2}}

\newcommand{\relD}[2]{\sim^{#1}_{#2}}

\newcommand{\actMf}[1]{\mathcal{#1}}  
\newcommand{\amSA}[1]{\actMf{SA}_{#1}}    
\newcommand{\amIS}{\actMf{IS}}    

\newcommand{\fX}{\mathop{\rtimes}}  

\newcommand{\proj}{\pi}  

\newcommand{\ISA}[1]{\cplI[\amSA{#1}]}    
\newcommand{\IIS}[1]{\cplI[\amIS^{#1}]}    

\newcommand{\PUfc}[1]{\widehat{#1}}   
\newcommand{\fchange}[2]{{#1}_{#2}}
\newcommand{\IISfc}[2]{\fchange{\IIS{#1}}{#2}}    

\newcommand{\kconc}[1]{\mathcal{R}_{#1}} 
\newcommand{\carrier}[1]{\mathit{carrier}_{#1}}
\newcommand{\Contention}{\mathit{Cont}}

\newcommand{\osp}[1]{\langle #1 \rangle}
\newcommand{\flip}[2]{\underline{#2}_{#1\strut}}

\newcommand{\Rconn}[1]{\mathrel{\frown_{#1}}}

\newcommand{\Value}{\mathit{Value}}

\newcommand{\Facet}{\mathsf{F}}

\newcommand{\coloring}{\chi}
\newcommand{\colorSet}{\Pi}

\definecolor{nodeblue}{RGB}{0,0,204}
\definecolor{nodered}{RGB}{255,204,204}
\newcommand{\nodeW}{{\circ}}
\newcommand{\nodeB}{{\color{nodeblue}\bullet}}
\newcommand{\nodeR}{{\color{nodered}\bullet}}

\newcommand{\viewOf}[1]{\mathit{view}_{#1}}

\newcommand{\PowerSet}[1]{2^{#1}}
\newcommand{\anglpair}[1]{\langle{#1}\rangle}

\newcommand{\abs}[1]{{\lvert}#1{\rvert}}

\begin{document}
\title{Proving Unsolvability of Set Agreement Task with\\ Epistemic $\mu$-Calculus} 
\author{
    Susumu Nishimura\thanks{%
    The author is supported by JSPS KAKENHI Grant Number 20K11678.}\\[0.5ex]
    {\small\texttt{susumu@math.kyoto-u.ac.jp}}
    \\[0.5ex]
    {\small Dept.\ Math, Kyoto University, Japan}
    }

\maketitle



\begin{abstract}
    This paper shows, in the framework of the \keywd{logical method}
    \cite{GoubaultLedentRajsbaum:InfComput20,DitmarschGoubaultLazicLedentRajsbaum21}, 
    the unsolvability of $k$-set agreement task    
    by devising a suitable formula of epistemic logic.
    The unsolvability of $k$-set agreement task is a well-known fact, which is a
    direct consequence of Sperner's lemma, a classic result from combinatorial topology.
    However, Sperner's lemma does not provide a good intuition for the unsolvability, 
    hiding it behind the elegance of its combinatorial statement. 
    The logical method has a merit that it can account for the reason of unsolvability
    by a concrete formula, but no epistemic formula for the general unsolvability 
    result for $k$-set agreement task has been presented so far. 
    
    We employ a variant of \keywd{epistemic $\mu$-calculus}, which extends the
    standard epistemic logic with distributed knowledge operators and
    propositional fixpoints, as the formal language of logic.
    With these extensions, we can provide an epistemic $\mu$-calculus formula
    that mentions higher-dimensional connectivity, which is essential
    in the original proof of Sperner's lemma, 
    and thereby show that $k$-set agreement tasks are not solvable even by 
    multi-round protocols. 
    Furthermore, we also show that the same formula applies to establish 
    the unsolvability for $k$-concurrency, a submodel of the 2-round protocol.


 \end{abstract} 


\section{Introduction}
\label{sec:intro}


A distributed \keywd{task} is a computational problem to 
be solved by a distributed system, which consists of 
$n+1$ mutually communicating processes.
A task specifies, for a given set of initial inputs to the processes,
what final output values are expected to be produced by the processes.
For instance, \keywd{$k$-set agreement}, which is
a fundamental distributed computing problem,
is specified as a task that produces a set of final outputs
satisfying the following properties.
\begin{description}
    \item[Agreement] At most $k$ different output values are decided by
    the $n+1$ processes;\\[-1.5em]
    \item[Validity] The output decided by each process must be one of the 
    $n+1$ input values given to the processes.
\end{description}
We say a task is solvable, 
if there exists a distributed algorithm, called a \keywd{protocol}, that conforms to the task specification. 
Throughout the paper, we argue distributed task solvability by 
wait-free protocols, 
assuming asynchronous distributed systems with crash-failure.

Research development in the last decades has found
the so called \emph{topological method} \cite{Herlihy:DCTopology}
as a means for discussing properties of distributed systems, especially 
unsolvability of distributed tasks, on the firm theoretical ground of combinatorial topology. 
In the topological method, 
a state of a distributed system of $n+1$ processes is modeled 
by an $n$-dimensional \keywd{simplex},  %
whose $n+1$ vertexes represent the local states of the $n+1$ processes.
A task or a protocol is defined as a function, called a 
\keywd{carrier map}, 
that relates each input simplex with a \keywd{(simplicial) complex}, i.e., 
a collection of output simplexes, which is intended to model the combination of
nondeterministic outputs by the $n+1$ processes.
The topological method  argues the unsolvability of a task by showing that
existence of a carrier map 
breaches certain topological invariant, 
where tools from combinatorial topology come into play.


\newtheorem{sperner}{Sperner's Lemma}
\renewcommand{\thesperner}{\unskip{}}

The topological method has been successful because of its cleanness, but 
it sometimes hides the reason of unsolvability behind the elegance of 
topological tools. For instance, one can cleanly prove the unsolvability of
the $k$-set agreement task  \cite{HerlihyShavit99,Herlihy:DCTopology}, admitting 
the following classic result from combinatorial topology. 
\begin{sperner}[\cite{Sperner28,Longueville13}]
    Suppose we are given a simplex $X$ whose vertexes are distinctly colored
    by a color set $S$.
    We say a coloring on vertexes of a subdivision of $X$ a
    \keywd{Sperner coloring}, if each vertex $v$ in the subdivision
    has the same color as a vertex of $Y$, where $Y$ is the minimal 
    face of $X$ that contains $v$. 
    Then, any subdivision of the simplex $X$ with Sperner coloring 
    contains an odd number of simplexes 
    that are fully-colored by $S$.
\end{sperner}
\noindent%
The unsolvability of $k$-set agreement task is a direct 
consequence of Sperner's lemma, but 
the statement of the lemma does not provide
a good topological intuition for the reason of unsolvability.

Goubault, Ledent, and Rajsbaum \cite{GoubaultLedentRajsbaum:InfComput20}
proposed \emph{logical method} as an alternative to the topological method.
They observed that a topological model of simplicial complex has an equivalent
counterpart of Kripke model, where 
the adjacency of simplexes in a complex is interpreted as 
the relation over Kripke states that is suitable for epistemic reasoning.
This allows us to establish the unsolvability of a task just by finding 
a formula of epistemic logic \cite{DitmarschHoekKooi:DELbook08}
that is valid in the Kripke model of the task but not valid in the model of the protocol.
Their method has a strong merit that such a formula can explicitly account for
the unsolvability. 


In its earliest development, 
the logical method has come with few concrete formulas, however.
Goubault et al.\  \cite{GoubaultLedentRajsbaum:InfComput20} 
provided such an epistemic formula for consensus task (i.e., 1-set agreement) 
and $N$-approximate agreement task but 
they left general $k$-set agreement as an open problem.
Soon later, Nishida \cite{Nishida:Msc20} devised a concrete epistemic logic formula
to show the unsolvability of $k$-set agreement and 
this formula is further refined  in
\cite{YagiNishimura:arXiv20}
to show the unsolvability under superset-closed adversary, which generalizes wait-freedom.
However, both studies solely concern 
single-round protocols, where each process in the system is allowed to 
communicate with the others by accessing the communication device only once. 
In contrast, the topological method 
works for multi-round protocols, 
where processes are allowed to communicate arbitrarily many times, 
as so usual in most of the distributed environments.

The aim of this paper is to provide a concrete formula 
that proves the unsolvability of $k$-set agreement task by 
multi-round protocols in the framework of the logical method
and thereby to express the reason for the unsolvability
in the formal language of epistemic logic.
The statement of Sperner's lemma is, however, not directly expressible
by epistemic logic, because there is no 
effective means to enumerate all relevant epistemic states in a Kripke model.
So we instead elaborate on the combinatorial structure that is 
exploited in the proof of Sperner's lemma \cite{Longueville13}.
For this, we need to extend the formal language of epistemic logic 
with the following three aspects:
\begin{description}
    \item[Epistemic $\mu$-calculus for connectivity.] 
    We extend the language of the epistemic logic 
    with propositional greatest fixpoint $\nu Z.\varphi$, which gives rise
    to the \keywd{epistemic $\mu$-calculus}
    \cite{FaginHalpernMosesVardi:book95,ShilovGaranina02}, 
    an epistemic variant of modal $\mu$-calculi \cite{Kozen83,BradfieldStriling:ModealMu07}.
    While epistemic modalities, such as the knowledge modality $\ModK{a}$,  
    in the standard epistemic logic 
    can only mentions adjacent epistemic states of finite proximity in a Kripke model,
    the epistemic $\mu$-calculus can express indirect adjacency of unbounded diameter, e.g., 
    transitive closure of adjacency relations, by means of greatest fixpoints.

    \item[Distributed knowledge for higher-dimensional connectivity.] 
    We further extend the epistemic $\mu$-calculus 
    by introducing distributed knowledge modality $\ModD{A}$, where $A$ is 
    a subset of the $n+1$ processes. 
    Topologically, while the knowledge modality $\ModK{a}$ concerns
    simplexes that are connected via a single common vertex, 
    the distributed knowledge modality $\ModD{A}$ \cite{FaginHalpernMosesVardi:book95,HalpernMoses90}
    concerns simplexes connected via a common simplex of an arbitrary dimension,
    e.g., a vertex (dimension~0 for $\abs{A}=1$), an edge (dimension~1 for $\abs{A}=2$), 
    a surface (dimension~2 for $\abs{A}=3$), etc.

    \item[Factual change of atomic propositions mentioning output decision.] 
    We carry out our development in a logical model that extends 
    the original one with 
    factual change of atomic propositions \cite{DitmarschGoubaultLazicLedentRajsbaum21}.
    The extended model provides an extra set of atomic formulas that 
    can mention the final output decision values made by processes. 
\end{description}

Sperner's lemma has been proven by examining the combinatorial structure
that is specific to a certain graph induced from Sperner coloring, 
where the graph nodes are the simplexes and the edges are 
connecting those adjacent simplexes, which share a common face. 
We paraphrase this argument on the graph structure 
into that on higher-dimensional connectivity of simplicial complex 
and construct a relevant logic formula, combining the greatest fixpoint 
and distributed knowledge modality, where 
the formula also needs to mention output decision values
in order to describe the constrains on coloring.
Specifically, assuming the solvability of $k$-set agreement, 
we show that a certain epistemic $\mu$-calculus formula entails the existence of
a cycle-free, ever-lasting connectivity path, which leads to a contradiction.

\subsection{Related work}
\label{subsec:relatedwork}

To the best of author's knowledge, 
there has been no proposal of a concrete epistemic logic formula that 
gives an account for the reason that multi-round protocols cannot solve $k$-set agreement task.
Recent developments on the logical method have shown that 
the standard epistemic logic, which employs solely knowledge modality, 
is rather weak for establishing unsolvability. In \cite{DitmarschGoubaultLazicLedentRajsbaum21}, 
a bisimulation technique is applied to show that no formula of the 
standard epistemic logic can establish the unsolvability of a certain task, 
which can be proven unsolvable by the topological method. 
Very recently, Hoshino \cite{Hoshino22} observed that, applying a simulation technique, 
the unsolvability of set agreement tasks by multi-round protocols cannot be established, 
even if the language of the logic is extended with distributed knowledge modality. 
This indicates that the formula proposed in \cite{Nishida:Msc20,YagiNishimura:arXiv20},
which makes use of distributed knowledge modality, 
is a peculiar instance that works solely for single-round protocols.


In \cite{DitmarschGoubaultLazicLedentRajsbaum21}, 
Ditmarsch et al.\ proposed to extend the logical model with factual change 
of atomic propositions \cite{GoubaultLedentRajsbaum:InfComput20}
and provided a generic formula of the extended logic that works
for virtually every task whose unsolvability can be established by a topological argument.
However, this generic formula contains no epistemic modality within
and therefore provides no hints for the reason of unsolvability.
In the present study, we also assume the logical model extended with factual change 
but present a formula that contains explicit epistemic contents 
expressing higher-dimensional connectivity.



From the perspective of topological method, 
it has been observed that
higher-dimensional connectivity, a.k.a., $k$-connectivity \cite{Herlihy:DCTopology,HerlihyRajsbaum00},
is essential for the unsolvability of $(k+1)$-set agreement task.
The present work expresses a similar (but topologically less general)
notion of higher-dimensional connectivity 
in the language of epistemic $\mu$-calculus.

\subsection{Organization}
The rest of this paper is organized as follows.
Section~\ref{sec:logicDC} reviews the topological model of distributed computing
and introduces the epistemic $\mu$-calculus for reasoning about distributed systems.
Section~\ref{sec:PUmodelkSet} gives the formal definition of product update models,
on which the unsolvability of $k$-set agreement task is argued. 
In Section~\ref{sec:IISunsolvable}, we prove the unsolvability 
with a concrete formula that encodes the essence of Sperner's lemma 
in the language of epistemic $\mu$-calculus.
We further discuss that the same method applies to a submodel of $k$-concurrency
in Section~\ref{sec:kconc}. 
%




\section{Logical Approach to Distributed Computing}
\label{sec:logicDC}

In what follows, $[\ell,m]$ denotes a range of nonnegative integers,
namely, the set $\{\ell,\ell+1, \ldots, m\}$.
We consider a
distributed system of $n+1$ processes, where
each individual process is distinguished by a unique process id, taken from the set 
$\colorSet=[0,n]$.
We say `process~$a$' to refer to the process identified
by the unique id $a\in\colorSet$. 
We also write $\Value$ to denote the set of possible initial input 
values given to each process. 
In the sequel, we assume $\Value=\colorSet$, 
renaming the input values as necessary. 

\subsection{Simplicial topology for distributed computing}
\label{subsec:simplicialDC}

In the topological theory of distributed computing \cite{Herlihy:DCTopology},
a particular state of an individual process is represented
by a pair $(a,v) \in \colorSet\times \Value$ to mean that 
a process~$a$ has a private value~$v$. 
The pair $(a,v)$ is regarded as a vertex of color~$a$ 
and we define the \keywd{coloring function} $\coloring$ on colored vertexes by $\coloring((a,v))=a$. 
A \keywd{simplex} is a collection of vertexes of pairwisely distinct colors
and a simplex consisting of $d+1$ vertexes is called a $d$-dimensional simplex.
A simplex $Y$ is called a \keywd{face} of a simplex $X$, if $Y\subseteq X$.
For a simplex $X$ and a color $a\in\coloring(X)$,
we define $\viewOf{a}(X)=v$ iff $(a,v)\in X$. 

A \keywd{simplicial complex} (or \keywd{complex} for short) $\cplC$ is 
a finite collection of simplexes closed under set inclusion: 
$X\in\cplC$ and $Y\subseteq X$ implies $Y\in\cplC$.
A maximal simplex in $\cplC$ is called a \keywd{facet}. 
Particularly, a distributed system of $n+1$ processes is modeled 
by a complex whose every facet is a simplex 
$\{(0,v_0), (1,v_1), \ldots, (n,v_n)\}$ of dimension $n$, 
where the $n+1$ vertexes are distintively colored through $0$ to $n$.
%
We write $\Facet(\cplC)$ (resp., $V(\cplC)$)
to denote the set of facets (resp., vertexes) contained in $\cplC$.
Suppose we are given complexes $\cplC$ and $\cplD$ colored by $\colorSet$. 
A \keywd{simplicial} map $\delta: V(\cplC) \to V(\cplD)$
is a color-preserving function on vertexes such that 
$\coloring(\delta(v))=\coloring(v)$
for every $v\in V(\cplC)$ and also 
$\delta(X)\in\cplD$ for every $X\in\cplC$.
We also define the cartesian product $\cplC\times\cplD$ as follows. 
For each pair of facets $X\in \Facet(\cplC)$ and $Y\in \Facet(\cplD)$,
we define $X\times Y$ as a simplex
    $X\times Y 
    = \{ 
        (a, (u,v)) \mid (a,u)\in X, (a,v)\in Y, a \in \colorSet
        \}$
so that each vertex in $X\times Y$ 
is given a pair of private values taken from the 
vertexes of matching color in $X$ and $Y$.
The cartesian product $\cplC\times\cplD$ 
is a complex determined by the set of facets
$\Facet(\cplC\times\cplD)=
\{X\times Y \mid
X\in\Facet(\cplC), Y\in\Facet(\cplD)\}$.

\subsection{Epistemic logic for distributed computing}
\label{subsec:epiLogicDC}



\subsubsection{The epistemic \texorpdfstring{$\mu$-calculus}{mu-calculus}}
\label{subsubsec:epistemiModel}


 
The syntax of the epistemic $\mu$-calculus formulas
considered in this paper is given by: 
\[
    \varphi::=
    p  \mid  \neg p \mid Z \mid \varphi\vee\varphi  \mid \varphi\wedge\varphi
    \mid \ModD{A}\varphi 
    \mid \nu Z.\varphi,
\]
where $p$ ranges over the set $\AtomProps$ of atomic propositions,
$Z$ ranges over the set $\PropVar$ of propositional variables, and
$A$ ranges over $\PowerSet{\colorSet}$. 

Notice that the formulas are restricted to \keywd{positive} ones:
Only atomic formulas are allowed to be negated by $\neg$.
In what follows, we say $\varphi$
is a propositional formula, 
if it does not contain any propositional variables, 
epistemic modalities, or greatest fixpoints.  
As usual convention, we write $\varphi \Rightarrow \psi$ 
for $\neg \varphi \vee \psi$, where $\varphi$ is a propositional formula
and $\psi$ is a positive formula.

\begin{figure}
    {
    \begin{align*}
        \semFormula{M}{\rho}{p} = {} & \{X\in S \mid p\in L(X)\}
        \qquad\quad
        \semFormula{M}{\rho}{Z} = {} \rho(Z)
        \qquad\quad
        \semFormula{M}{\rho}{\neg p} = {}  
        S\setminus \semFormula{M}{\rho}{p}  
        \\
        \semFormula{M}{\rho}{\varphi_1\vee\varphi_2} = {}  &
        \semFormula{M}{\rho}{\varphi_1} \cup
        \semFormula{M}{\rho}{\varphi_2}
        \qquad\qquad 
        \semFormula{M}{\rho}{\varphi_1\wedge\varphi_2} = {}  
        \semFormula{M}{\rho}{\varphi_1} \cap
        \semFormula{M}{\rho}{\varphi_2}
        \\
        \semFormula{M}{\rho}{\ModD{A} \varphi} = {} & 
        \{X\in S \mid Y\in \semFormula{M}{\rho}{\varphi} 
        \text{ for every $Y$ such that $Y\relD{}{A}X$} \}
        \\
        \semFormula{M}{\rho}{\nu Z.\varphi} = {} & 
        \bigcup 
        \{ S'\in \PowerSet{S} \mid S'\subseteq \semFormula{M}{\rho[S'/Z]}{\varphi} \}
    \end{align*}
    }\vspace*{-2em}
    \caption{Kripke semantics of the epistemic $\mu$-calculus}
    \label{fig:semantMuCalc}
\end{figure}

A \keywd{Kripke frame} is a pair $\anglpair{S,\sim}$,  where
 $S$ is the set of (epistemic) \keywd{states} and 
$\sim$ is a family $\{\sim_a \subseteq S\times S\mid a\in\colorSet\}$ of 
\keywd{indistinguishability relations}, with
each $\sim_a$ being an equivalence relation over $S$.
The Kripke model $M = \anglpair{S,\sim,L}$  augments the Kripke frame with 
a function $L: S\to \PowerSet{\AtomProps}$, which assigns, 
for each $X\in S$, 
a set $L(X)$ of atomic propositions that are true at the state $X$.
For a subset $A$ of $\colorSet$, 
we write $\relD{}{A}$ for a derived equivalence relation defined by 
$X \relD{}{A} Y$ iff $X \sim_a Y$ for all $a\in A$. 

The formal semantics of a formula $\varphi$ is given 
by the set $\semFormula{M}{\rho}{\varphi}$ of states at which $\varphi$ is satisfied,
where $\rho: \PropVar \to \PowerSet{S}$ gives an \keywd{interpretation}
of propositional variables that occur free in $\varphi$. 
Figure~\ref{fig:semantMuCalc} gives the definition of the semantics
by induction on the structure of formula $\varphi$.
In addition to propositional formulas, 
the epistemic $\mu$-calculus provides an
epistemic modality $\ModD{A}\varphi$, which is intended to express
the distributed knowledge 
asserting ``the group $A$ of processes know $\varphi$.''
It further provides greatest fixpoint $\nu Z.\varphi$, which denotes
the greatest solution satisfying the equation $Z = \varphi$ on the propositional variable $Z$
\cite{FaginHalpernMosesVardi:book95,ShilovGaranina02}.\unskip\footnote{%
    The least fixpoint $\mu Z.\varphi$, which is logically equivalent to $\neg \nu Z.\neg \varphi$,
    is missing from the present logic, because it is not a positive formula.}

%

In what follows, we write $M, X \models_{\rho} \varphi$
iff $X\in \semFormula{M}{\rho}{\varphi}$, to mean
that, $\varphi$ is \keywd{satisfied} at a particular state $X$ of Kripke model $M$, 
under an interpretation $\rho$.  
In particular, we write $M, X \models \varphi$, when 
$\varphi$ is a \keywd{closed} formula, i.e., $\varphi$ contains no free occurrences of 
propositional variables. 
We also write $M \models \varphi$ to mean 
a closed formula $\varphi$ is \keywd{valid}, that is,
$M, X \models \varphi$ holds at every state $X\in S$ of the Kripke model $M$. 

\subsubsection{Simplicial model: a Kripke model induced from simplicial complex}
\label{subsubsec:simplModelCpl}

The topological structure of a complex can be turned into 
a Kripke model, called a \keywd{simplicial model} \cite{GoubaultLedentRajsbaum:InfComput20}.
\begin{definition}
    \label{def:simplicialModel}
    A complex $\cplC$ induces a Kripke model  
    $M=\anglpair{\Facet(\cplC),\sim,L}$, called a \keywd{simplicial model}, where 
    \begin{itemize}
        \item the set of (epistemic) states is the set of facets $\Facet(\cplC)$ of $\cplC$,
        \item $\sim$ is a family of relations $\{\sim_a 
        \mid a\in \colorSet\}$ over $\Facet(\cplC)$, 
        where each $\sim_a$, called an
        \keywd{indistinguishability relation}, is an equivalence relation
        defined by  $X \sim_a Y \text{ iff } a \in \coloring(X\cap Y)$, 
        and
        \item $L: \Facet(\cplC)\to\PowerSet{\AtomProps}$,
        where $\AtomProps= \{ \Pinput{a}{v} \mid a\in\colorSet, v\in\Value\}$, 
        is a function defined by $L(X) = \{\Pinput{a}{v} \mid (a,v)\in X\}$.
    \end{itemize}
\end{definition}

By abuse of notation, we write $\cplC$ to refer to both 
the complex $\cplC$ itself and the simplicial model $M=\anglpair{\cplC,\sim,L}$
induced from it.
Simplicial models indeed have the identical structure 
with local proper Kripke models \cite{GoubaultLedentRajsbaum:InfComput20}, 
up to categorical equivalence 
under a suitable categorical setting.



\begin{figure} 
    \hfil%
    \begin{minipage}[b]{.4\linewidth} 
        \centering
        \includegraphics[scale=0.66]{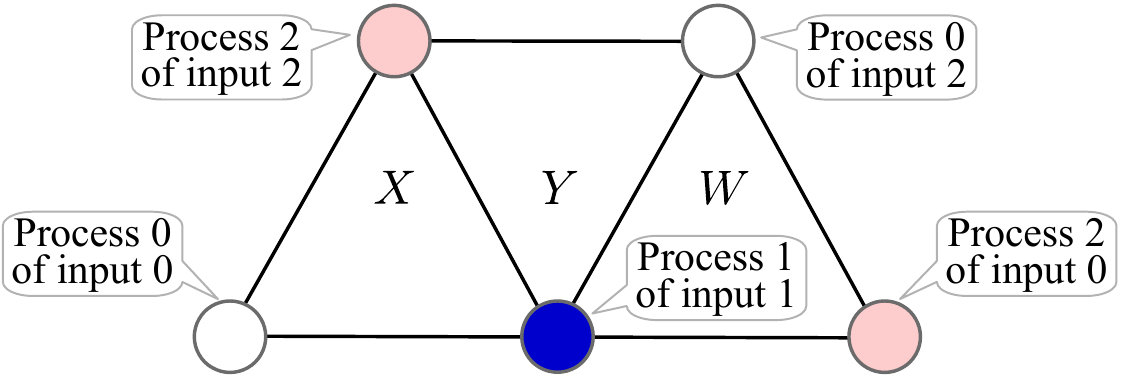}
        \subcaption{A simplical complex $\cplC$}
        \label{subfig:complexModel}
    \end{minipage}%
    \hfil%
    \hfil%
    \begin{minipage}[b]{.42\linewidth}
        \centering
        \includegraphics[scale=0.72]{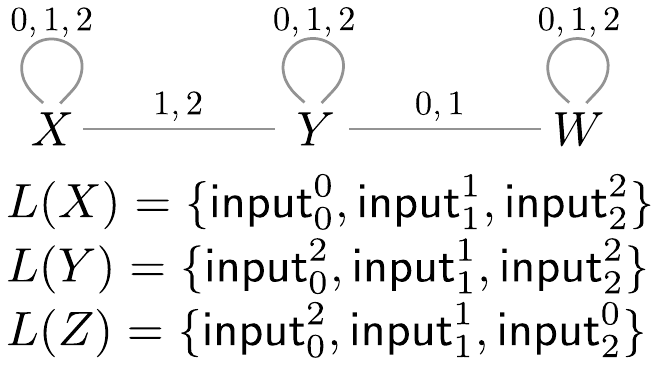}
        \subcaption{The corresonding Kripke model $\anglpair{\Facet(\cplC),\sim,L}$}
        \label{subfig:KripkeModel} 
    \end{minipage} 
    \caption{A topological model and its corresponding Kripe model for a 3~process system}
    \label{fig:simpModel} 
\end{figure}

Fig.~\ref{fig:simpModel} illustrates how a complex
is interpreted as a Kripke model. 
Fig.~\ref{subfig:complexModel} gives a 2-dimensional complex
consisting of three different facets, which  
models a 3~process system of process ids $\colorSet=\{0,1,2\}$. 
Each vertex receives its own private input as indicated in the figure
and each facet has 3 vertexes, as designated by 
vertexes of different colors $\nodeW$, $\nodeB$, and $\nodeR$, corresponding to
processes~$0$, $1$, and $2$, respectively. 
(Throughout the rest of the paper, we will follow this coloring convention.)
Fig.~\ref{subfig:KripkeModel} shows the corresponding simplicial Kripke model. 
The Kripke frame is depicted by an undirected graph, where 
each graph node stands for a state, corresponding to a facet in $\cplC$, 
and each graph edge labeled by $a\in\colorSet$ stands for a pair of nodes
related by $\sim_a$. 
(As such, a simplicial model can be pictorized by an undirected graph 
with self-cycles, because epistemic knowledge is modeled by an 
\textbf{S5} Kripke model, in which 
$\sim_a$ is an equivalence relation for each $a\in\colorSet$.)
The function $L$ determines 
the input value given to each vertex in each facet by the set of atomic formulas
of the form $\Pinput{a}{v}$, as in the figure.

Though the epistemic logic primarily concerns epistemic information in Kripke models, 
geometric interpretation on the original complexes, 
from which the simplicial models are induced, is insightful 
for the later development in this paper. 
%
%
%
A pair of facets are related by $\sim_a$ in an induced simplicial model
$\anglpair{\Facet(\cplC),\sim,L}$, 
as it immediately follows from the definition, 
if they share a common vertex of color~$a$ in the complex $\cplC$. 
Similarly, the derived relation 
$\relD{}{A}$ relates a pair of facets if they share a common face $U$ %
such that $\coloring(U) = A$. For example, in the complex of
Fig.~\ref{subfig:complexModel}, 
$X \relD{}{\{1,2\}} Y$ holds, because they share a 1-dimensional simplex (i.e., an edge
that connects $\nodeB$ and $\nodeR$).
In contrast, $X$ and $W$ are not related by  
$\relD{}{\{1,2\}}$ but by 
$\relD{}{\{1\}}$ via the common 0-dimensional simplex (i.e., vertex $\nodeB$).\footnote{%
    The distributed knowledge operator $\ModD{\{a\}} \varphi$ for a singleton set of processes
    is known as the knowledge operator $\ModK{a} \varphi$. 
    The knowledge operator is more popular in the literature of epistemic logic 
    but is omitted in the present logic, as it is 
    a special form of the distributed knowledge.}
Therefore $\cplC, X \models \ModD{\{1,2\}}\Pinput{2}{2}$
is satisfied because $\relD{}{\{1,2\}}$ relates $X$ with itself and $Y$,
while $\cplC, X \not\models \ModD{\{1\}}\Pinput{2}{2}$
because $X \relD{}{\{1\}} W$ but
$\cplC, W \not\models \Pinput{2}{2}$.

The greatest fixpoint provides extra power in expressing epistemic property such as
the common knowledge and the common distributed knowledge \cite{FaginHalpernMosesVardi:book95,BatlagSmets20}.
The common knowledge, written $\ModC{A} \psi$ for a subset $A$ of $\colorSet$, 
can be defined as the greatest fixpoint $\nu Z.(\psi \wedge \bigwedge_{a\in A}\ModD{\{ a\}} Z)$,
which can be understood as an infinite conjuction in repeated unfolding
$\psi \wedge \bigwedge_{a\in A}\ModD{\{ a\}} 
    \Bigl(\psi \wedge \bigwedge_{a\in A}\ModD{\{ a\}} 
        \bigl(\psi \wedge \bigwedge_{a\in A}\ModD{\{ a\}} (\ldots)\bigr)\Bigr)$.
This defines the common knowledge $\ModC{A}\psi$ as a property that 
$\psi$ is satisfied at every state in the reach of 
the trasitive closure of the union relation $\bigcup_{a\in A} \sim_a$. 
For example, in the complex $\cplC$ of Fig.~\ref{subfig:complexModel}, 
$\cplC, X\not\models \ModC{\{0,2\}}\Pinput{2}{2}$ 
but $\cplC, X\models \ModC{\{2\}}\Pinput{2}{2}$, 
because $X \sim_2 Y \sim_0 W$ and hence $W$ is in the reach of the former 
but not the latter.



In the subsequent sections, task solvability will be defined with the notion
of morphisms.
Let $\anglpair{\Facet(\cplC),\sim,L}$ and $\anglpair{\Facet(\cplD),\sim',L'}$
be simplicial models induced from complexes $\cplC$ and $\cplD$, respectively. 
We say a function $\delta: V(\cplC)\to V(\cplD)$ is a \keywd{morphism}, if 
$\delta$ is a color-preserving simplicial map and furthermore
$L(X)=L'(\delta(X))$ holds 
for every $X\in\Facet(\cplC)$. 


The knowledge gain theorem, 
which is essential for showing task unsolvability in the logical method, 
is conservatively extended to allow the additional logical connectives, i.e.,
greatest fixpoints and distributed knowledge modalities, in 
the epistemic $\mu$-calculus.
%
\begin{restatable}[knowledge gain]{theorem}{knowledgeGain}\label{th:knowledgegain}
    Suppose 
    $\cplC=\anglpair{\Facet(\cplC),\sim,L}$ and 
    $\cplD=\anglpair{\Facet(\cplD),\sim',L'}$ are simplicial models
    and $\delta$ is a morphism from $\cplC$ to $\cplD$.  
    Then, for any state $X\in \Facet(\cplC)$ and closed positive formula $\varphi$,
    $\cplD,\delta(X)\models\varphi$ implies $\cplC,X\models\varphi$.
\end{restatable}


\section{Product Update Models for \texorpdfstring{$k$-Set}{k-Set} Agreement}
\label{sec:PUmodelkSet}

In the logical method, distributed computation is modeled 
by \keywd{product updates}, which have been used for defining
update on the knowledge model
in dynamic epistemic logic (DEL) \cite{BaltagMossSolecki16,DitmarschHoekKooi:DELbook08}.
This section presents concrete instances of product updates 
that are relevant for discussing $k$-set agreement.

\subsection{Product updates and task solvability}
\label{subsec:PUsolv}

Suppose $\cplC$ and $\cplD$ are complexes. 
A product update model is derived from a complex $\cplC[\cplD]$, 
which is a subcomplex of the the cartesian product $\cplC\times \cplD$
(i.e., $\Facet(\cplC[\cplD])\subseteq \Facet(\cplC\times \cplD)$). 
The subcomplex $\cplC[\cplD]$ can be regarded as a specification of distributed computation: 
$X \times Y \in \Facet(\cplC[\cplD])$ implies that 
$Y\in \Facet(\cplD)$ is one of possible outputs for an input $X\in \Facet(\cplC)$. 
In other words, $\cplC[\cplD]$ is a binary relation over $\Facet(\cplC)\times\Facet(\cplD)$,
which encodes a one-to-many mapping from input facets to output facets. 
(Formally, such a subcomplex $\cplC[\cplD]$ is determined by an epistemic formula,
called a \keywd{precondition}, defined for each $Y\in \Facet(\cplD)$. 
See \cite{GoubaultLedentRajsbaum:InfComput20,DitmarschGoubaultLazicLedentRajsbaum21}
for the formal definition.)

Let $\cplC[\cplD]$ be a subcomplex of $\cplC\times \cplD$ and
$\anglpair{\Facet(\cplC),\sim^{\cplC},L}$ 
be the simplicial model induced from $\cplC$.
A \keywd{product update model} is a simplicial model 
$\anglpair{\Facet(\cplC[\cplD]),\sim^{\cplC[\cplD]},L'}$, 
where $\anglpair{\Facet(\cplC[\cplD]),\sim^{\cplC[\cplD]}}$ is the Kripke frame induced from $\cplC[\cplD]$
and $L'$ is defined by $L'(X\times Y) = L(X)$ for every $X\times Y\in \Facet(\cplC[\cplD])$.

In what follows, we write $\cplI$ to denote the input complex, 
i.e., a complex whose set of facets is 
$\Facet(\cplI) = \bigl\{ \{(0,v_0),\ldots,(n,v_n)\} \mid v_0,\ldots,v_n\in \colorSet\bigr\}$.
To argue if a task is solvable by a protocol in the logical method, 
we model the task and the protocol by 
product update models $\cplI[\complexf{T}]$ and 
by $\cplI[\complexf{P}]$, respectively, 
where $\complexf{T}$ is the output complex for the task and
$\complexf{P}$ is that for the protocol. 
The task solvability is defined as below, via the notion of morphism, as follows. 
\begin{definition}[task solvability\cite{GoubaultLedentRajsbaum:InfComput20}] \label{def:PUtasksolavbility}
    A task $\cplI[\complexf{T}]$ is solvable by a protocol $\cplI[\complexf{P}]$
    iff there exists a morphism $\delta:\cplI[\complexf{P}]\to \cplI[\complexf{T}]$ 
    such that $\proj_{\cplI}\circ \delta = \proj_{\cplI}$,
    where $\proj_{\cplI}$ is the first projection, i.e., 
    $\proj_{\cplI}((X\times Y))=X$.
\end{definition}

The product update model presented above, however, does not suffice for 
showing the unsolvability of $k$-set agreement, as 
we have mentioned in Section~\ref{sec:intro}. 
We need to extend it with factual change \cite{DitmarschGoubaultLazicLedentRajsbaum21}
so that the output decision values can be mentioned by atomic propositions.


Let $\cplI[\complexf{T}]$ denote the product update model 
$\anglpair{\Facet(\cplI[\complexf{T}]), \sim^{\cplI[\complexf{T}]}, L}$ 
for the task and $\cplI[\complexf{P}]$ denote the product update model 
$\anglpair{\Facet(\cplI[\complexf{P}]), \sim^{\cplI[\complexf{P}]}, L'}$ 
for the protocol.
Let us assume an augmented set of atomic propositions
$\PUfc{\AtomProps} = \AtomProps \cup 
\{\Pdecide{a}{v} \mid a\in\colorSet, v\in\Value\}$,
where $\Pdecide{a}{v}$ is an atomic proposition asserting that
process~$a$ decides $v$ as its output value. 
Then 
the product update model with factual change for the task, written 
$\PUfc{\cplI[\complexf{T}]}$,  is a modified product update model
$\anglpair{\Facet(\cplI[\complexf{T}]), \sim^{\cplI[\complexf{T}]}, \PUfc{L}}$, 
where $\PUfc{L}$ augments the assignment of true atomic propositions by
$\PUfc{L}(X\times Y) = \{\Pinput{a}{v} \mid (a,v)\in X  \} \cup \{\Pdecide{a}{d} \mid (a,d)\in Y  \}$.
Moreover, for each given morphism $\delta: \cplI[\complexf{P}] \to \cplI[\complexf{T}]$, 
the product update model with factual change for the protocol, written 
$\fchange{\cplI[\complexf{P}]}{\delta}$, is a modified product update model
$\anglpair{\Facet(\cplI[\complexf{P}]), \sim^{\cplI[\complexf{P}]}, \PUfc{L}'}$,
where $\PUfc{L}'(X\times Y) = \PUfc{L}(\delta(X\times Y))$
for every $X\times Y \in \Facet(\cplI[\complexf{P}])$.

%

To show the unsolvability of $k$-set agreement, 
we will resort to the following property of product update with factual change.
(This claim follows from Theorem~\ref{th:knowledgegain} and 
certain properties of  product updates with factual change.
For the details, see the discussion in the proof of Theorem~19 in
\cite{DitmarschGoubaultLazicLedentRajsbaum21}.)
\begin{theorem} \label{th:PUlogical} 
    Let $\cplI[\complexf{T}]$ and $\cplI[\complexf{P}]$ be the product update models 
    of a task and a protocol, respectively. Suppose 
    there exists a morphism $\delta: \cplI[\complexf{P}] \to \cplI[\complexf{T}]$. 
    Then, $\PUfc{\cplI[\complexf{T}]}, \delta(X) \models \varphi$ implies 
    $\fchange{\cplI[\complexf{P}]}{\delta}, X\models \varphi$, for any 
    $X\in\Facet(\fchange{\cplI[\complexf{P}]}{\delta})$ and positive formula $\varphi$. 
\end{theorem}



\subsection{Product update model 
    \texorpdfstring{$\ISA{k}$}{I[SAk]}
    for \texorpdfstring{$k$-set}{k-set} agreement task}
\label{subsec:PUmodelkSA}


The output complex of $k$-set agreement task is specified by 
the set of facets $\Facet(\cplO_k)= \bigl\{ \{(0,d_0), \ldots, (n,d_n)\} \bigm\vert
\abs{\{d_0, \ldots, d_n\}}\leq k, d_0, \ldots, d_n\in\colorSet \bigr\}$.
Hence the product update model $\ISA{k}$ is given by 
$\langle\Facet(\ISA{k}),$ $\sim^{\ISA{k}},L\rangle$, 
where 
$\Facet(\ISA{k}) = \{ I\times O \mid I\in\Facet(\cplI), O\in\Facet(\cplO_k), 
    \{v \mid (a,v)\in I\}\supseteq \{d \mid (a,d)\in O\} \}$
is the set of facets respecting the validity condition.


This product update model can be extended to the one with factual change
$\PUfc{\ISA{k}} = \langle \Facet(\ISA{k}), \sim^{\ISA{k}}$, $\PUfc{L} \rangle$,
where
$\PUfc{L}(I\times O) = \{\Pinput{a}{v} \mid (a,v)\in I\} 
\cup \{\Pdecide{a}{d} \mid (a,d)\in O\}$
for each $I\times O \in \Facet(\ISA{k})$.  


\begin{figure}[t] \centering
    \includegraphics[width=0.48\textwidth]{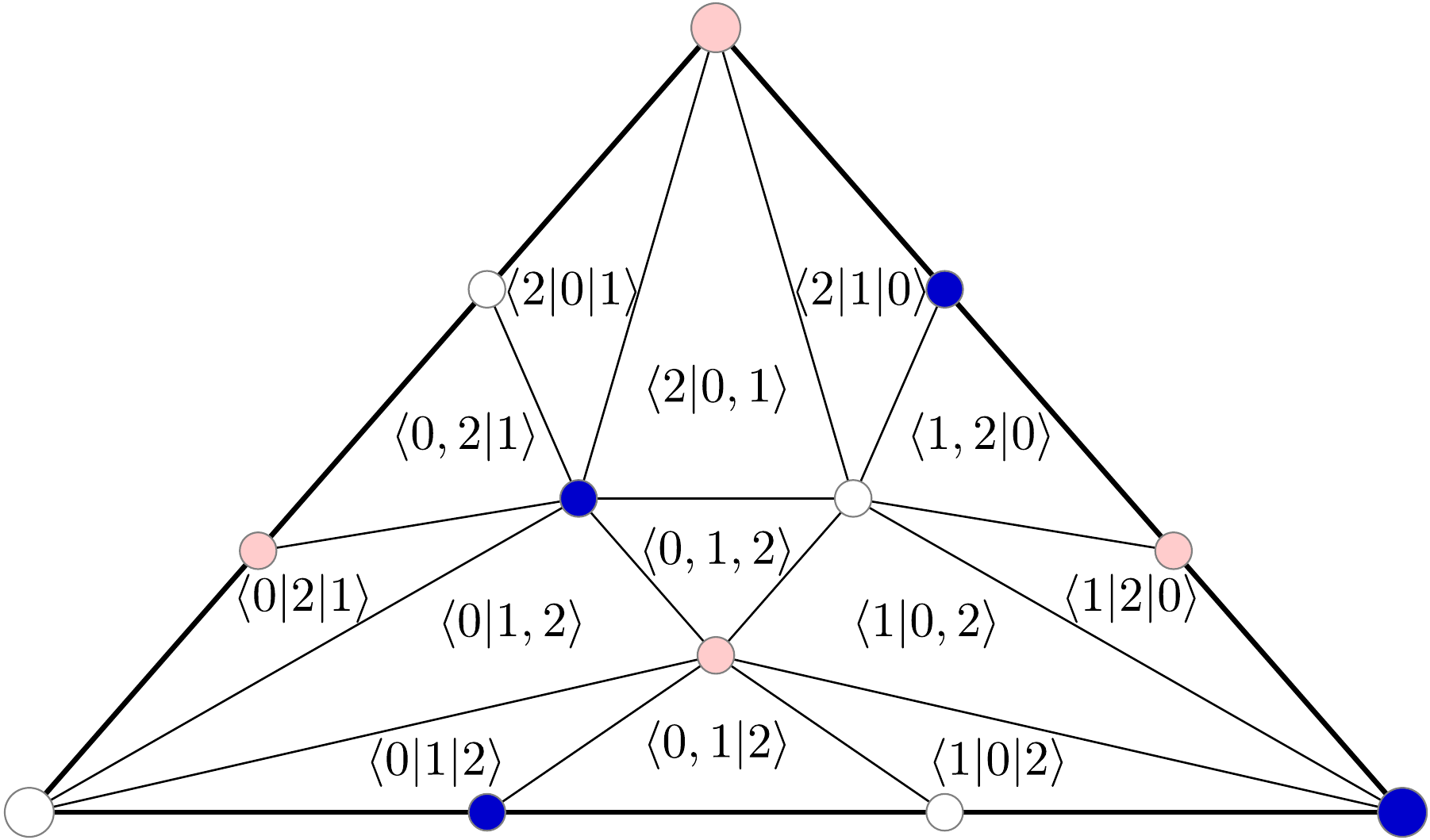}
    \caption{Standard chromatic subdivision of a 2-dimensional simplex
    with facets being labeled with ordered set partitions.}
    \label{fig:IS1osp}
\end{figure}

\subsection{Product update model \texorpdfstring{$\IIS{m}$}{I[ISm]}
    for multi-round \texorpdfstring{$m$}{m}-iterated immediate snapshot protocol}
\label{subsec:PUmodelIIS}

In this paper, we 
assume the read-write shared memory model as a device for inter-process communication. 
The topological properties of the 
read-write shared memory distributed systems are well understood.
%
As for wait-free solvability, 
the read-write shared memory model
is equivalent to the iterated immediate snapshot protocol 
\cite{BorowskyGafni:PODC93,GuerraouiKuznetsov:AlgConcurrentSys,Nishimura:OPODIS17},
a multi-round protocol in which processes communicate via
immediate snapshot objects. 

In each single round execution of the immediate snapshot protocol,
the $n+1$ processes are arranged into an arbitrary sequence of 
concurrency classes $A_1, \ldots, A_r$ ($r\geq 1$), 
a disjoint partition of $\colorSet$ by nonempty sets. 
According to the arranged order, 
processes take snapshot: Each process writes its own input value in a shared memory 
and collects the values that have been written so far by other processes
and itself. The set of values collected by a snapshot operation is called a \keywd{view}. 
Processes belonging to the same concurrency class $A_q$
witness the same view, i.e., the set of values 
written by the processes of ids $\bigcup_{i=1}^{q}A_i$.
In the $m$-iterated multi-round immediate snapshot protocol, 
each process repeatedly takes snapshots, 
using the view of $i$-th round as the input to the 
subsequent $(i+1)$-th round.


%
The virtue of the immediate snapshot model is 
that each snapshot operation amounts to a topological operation of subdivision,
called the \keywd{(iterated) standard chromatic subdivision} \cite{HerlihyShavit99,Herlihy:DCTopology}.
Figure~\ref{fig:IS1osp} illustrates the standard chromatic subdivision  
of a 2-dimensional simplex,  where 
each facet of the subdivision is distinctly labeled
by a sequence of concurrency classes, or an \keywd{ordered set partition}. 
We write $\osp{A_1{\mid} A_2 {\mid}\cdots{\mid} A_r}$ to denote an 
ordered set partition, where $A_1, A_2, \ldots, A_r$ ($r>0$) is 
a sequence of nonempty subsets such that $\colorSet=\bigcup_{i=1}^r A_i$ \cite{Kozlov12}. 
(In what follows, we omit curly braces in an ordered set partition, say, 
$\osp{0,1{\mid} 2 {\mid} 3}$ instead of 
$\osp{\{0,1\}{\mid} \{2\} {\mid} \{3\}}$.)

Let $\gamma, \gamma'$, etc.\ range over ordered set partitions.
For a facet $X\in \Facet(\cplI)$, 
we write $X\fX \gamma$ to denote
a facet of the standard chromatic subdivision of $X$
that is uniquely identified by an ordered set partition $\gamma$.  
The snapshot view of a process~$a$ 
in a facet $X\fX\osp{A_1{\mid} A_2 {\mid}\cdots{\mid} A_r}$ of the subdivision is defined by: 
\begin{align*}
    \viewOf{a}(X\fX\osp{A_1{\mid} A_2 {\mid}\cdots{\mid} A_r}) = {} &
    \{(b,v)\in X \mid b\in \textstyle\bigcup_{i=1}^q A_i\}, 
    \text{ where $a\in A_q$.}
\end{align*}
In other words, a facet 
$X\fX\osp{A_1{\mid} A_2 {\mid}\cdots{\mid} A_r}$ 
denotes a set of vertexes $\bigl\{ \bigl(a,\viewOf{a}(X\fX\langle A_1{\mid} A_2 {\mid}\cdots$ ${\mid} A_r\rangle)\bigr)
\mid a\in \colorSet\bigr\}$.
We also write $X\fX \gamma_1 \fX \cdots \fX \gamma_m$ ($m\geq 1$) to denote
a facet of subdivision by an $m$-iterated immediate snapshot, 
where each ordered set partition $\gamma_i$ corresponds to the $i$-th snapshot,
with $\fX$ associating to left.
The snapshot view is defined by induction on $m$. 

The product update model 
of the $m$-iterated immediate snapshot protocol is  
given by 
$\IIS{m} = \langle\Facet(\IIS{m}),$ $ \sim^{\IIS{m}}, L'\rangle$,
where the set of facets is $\Facet(\IIS{m}) = 
\{ X \fX \gamma_1 \fX \cdots \fX \gamma_m \mid 
X\in\Facet(\cplI), 
\gamma_1,\ldots,\gamma_m$ are ordered set partitions of $\colorSet\}$\footnote{%
    Formally, a facet 
    in $\IIS{m}$ must be written
    in the form  $X \times (X \fX \gamma_1 \fX \cdots \fX \gamma_m)$, 
    with duplicates of $X$, but for brevity
    we write it in the abridged notation $X \fX \gamma_1 \fX \cdots \fX \gamma_m$.
    An alternative definition of product update proposed in 
    \cite{VelazquezCastanedaRosenblueth21} can circumvent this kind of duplication, 
    but we stick to the original definition.}
and
$L'(X \times \gamma_1 \fX \cdots \fX \gamma_m) = \{\Pinput{a}{v} \mid (a,v)\in X\}$
for each  $X \times \gamma_1 \fX \cdots \fX \gamma_m \in \Facet(\IIS{m})$.

In \cite{Kozlov12}, using ordered set partition, Kozlov 
gave the combinatorial description of incidence of facets contained in
the standard chromatic subdivision of a simplex.
This entails that the facets in $\IIS{}$ 
admit the following incidence relation, which 
is suitable for the subsequent discussion. 



%
Suppose $0\leq d\leq n$ and $b\in[0,d]$.  
Let $A =[0,d]\setminus \{ b\}$  and $X \fX \gamma\in \Facet(\IIS{})$ be a facet,
where $X\in\Facet(\cplI)$ and 
$\gamma =\osp{A_1 {\mid} A_2{\mid} \cdots{\mid} A_r 
{\mid} d+1 {\mid} d+2  {\mid}\cdots{\mid} n}$ ($r>0$) 
is an ordered partition of $[0,n]$. Then, 
$X\fX \gamma \relD{}{A} Y \fX \osp{A_1' {\mid} A_2'{\mid} \cdots{\mid} A_{r'}' 
{\mid} d+1  {\mid}\cdots{\mid} n}$ holds if and only if either of the following holds. 
\begin{itemize}
    \item $A_r=\{b\}$, 
    $\osp{A_1' {\mid} A_2'{\mid} \cdots{\mid} A_{r'}' 
    {\mid} d+1  {\mid}\cdots{\mid} n}=\gamma$,  
    and $X \relD{}{A} Y$.

    \item $A_r\neq\{b\}$,
    $X=Y$, and $\osp{A_1' {\mid} A_2'{\mid} \cdots{\mid} A_{r'}' 
    {\mid} d+1  {\mid}\cdots{\mid} n}$ is either $\gamma$ or $\flip{A}{\gamma}$,
    where 
    \begin{align*}
        \flip{A}{\gamma} = {}&
        \begin{cases}
            \osp{A_1 {\mid} 
            \cdots{\mid} A_{r-1}{\mid} A_{r}{\setminus}\{b\}       
            {\mid} \,b\, {\mid} 
            d+1  {\mid}\cdots{\mid} n}
            & \text{if $b\in A_r$, $\abs{A_r}>1$,}
            \\
            \osp{A_1 {\mid} 
            \cdots{\mid} A_{s-1}{\mid} A_{s}{\setminus}\{b\}{\mid} \,b\, {\mid}A_{s+1}{\mid}      
            \cdots{\mid} A_r{\mid} 
            d+1  {\mid}\cdots{\mid} n}
            & \text{if $b\in A_s$, $\abs{A_s}>1$  ($s< r$),}
            \\
            \osp{A_1 {\mid} 
            \cdots{\mid} A_{s-1}
            {\mid} A_{s+1}{\cup}\{b\}{\mid}         
            \cdots{\mid} A_r{\mid} 
            d+1  {\mid}\cdots{\mid} n}
            & \text{if $A_s=\{ b\}$ ($s< r$).}
        \end{cases}
    \end{align*}
\end{itemize}

This extends to the multi-round model $\IIS{m}$ as follows.
\begin{proposition} \label{prop:opsincident}
    Suppose $0\leq d\leq n$ and $b\in[0,d]$.  
    Let $A =[0,d]\setminus \{ b\}$  and 
    $X\fX\gamma_1\fX \cdots\fX \gamma_m \in \Facet(\IIS{m})$, where 
    each $\gamma_i$ is an ordered set partition of the restricted form
    $\osp{A_1 {\mid} A_2{\mid} \cdots{\mid} A_{r_i} {\mid} d+1  {\mid}\cdots{\mid} n}$.
    Then, 
    $X\fX\gamma_1\fX \cdots\fX \gamma_m \relD{}{A}
    Y\fX\gamma_1'\fX \cdots\fX \gamma_m'$, 
    where each $\gamma_i'$ is an ordered set partition of the same restricted form,  
    holds if and only if either of the following hold.
    \begin{itemize}
        \item $A_{r_i}=\{b\}$ and $\gamma_i'=\gamma_i$ holds for every $i$ and also
        $X \relD{}{A} Y$.

        \item $A_{r_i}\neq\{b\}$ for some $i$, $X=Y$, and 
        $\gamma_i'=\gamma_i$ holds for every $i$ except for 
        the minimum index $j$ such that $A_{r_j}\neq \{b\}$
        and $\gamma_j'$ is either $\gamma_j$ or $\flip{A}{\gamma_j}$.    
    \end{itemize}
\end{proposition}    



\subsection{Logical properties of product update models 
    \texorpdfstring{$\PUfc{\ISA{k}}$}{I[SAk]} and \texorpdfstring{$\IIS{m}$}{I[ISm]}}
\label{subsec:logicalProp}

In the sequel, let $\sigma, \tau$, etc.\ range over facets of a product update model. 

In both $\PUfc{\ISA{k}}$ and $\IIS{m}$, 
the following formula $\IFUN$ is valid. 
\begin{align} \label{eq:IFUN}
    \IFUN = {} & \textstyle
    \bigwedge_{a\in\colorSet}
    \bigl(\bigwedge_{i,j\in\colorSet,i\neq j} 
    \neg(\Pinput{a}{i}\wedge\Pinput{a}{j})
    \wedge
    \bigvee_{i\in\colorSet}\Pinput{a}{i}
    \bigr)
\end{align}
This formula is intended to mean that each process~$a$ is given a unique input value.

In $\PUfc{\ISA{k}}$, the following formulas concerning output values are also valid.
\begin{align}
    \OFUN = {} &  \textstyle
    \bigwedge_{a\in\colorSet}
    \bigl(
        \bigwedge_{d,e\in\colorSet,d\neq e} \neg(\Pdecide{a}{d}\wedge\Pdecide{a}{e})
        \wedge
        \bigvee_{d\in\colorSet}\Pdecide{a}{d}
    \bigr)
    \label{eq:OFUN}
    \\
    \VALID = {} & \textstyle
    \bigwedge_{a\in\colorSet}
    \bigwedge_{d\in\colorSet}
    \bigl(
        \Pdecide{a}{d} \Rightarrow \bigvee_{b\in\colorSet} \Pinput{b}{d}
    \bigr)
    \label{eq:VALID}
    \\
    \AGREE{k} = {} & \textstyle
    \bigvee_{A \subseteq \colorSet, 0<\abs{A}\leq k}
    \bigwedge_{a\in\colorSet} \bigvee_{d\in A} \Pdecide{a}{d}
    \label{eq:AGREE}
    \\
    \KNOW = {} & \textstyle
    \bigwedge_{A \subseteq \colorSet}
    \bigwedge_{a\in A} 
    \bigwedge_{d\in \colorSet} 
    \bigl(\Pdecide{a}{d} \Rightarrow \ModD{A} \Pdecide{a}{d}\bigr)
    \label{eq:KNOW}
\end{align}
The formula $\OFUN$ 
means that each process~$a$ 
decides a unique output value.
%
The formula $\VALID$ expresses the validity condition:
Any output value $d$ must have been input to some of the processes. 
The formula $\AGREE{k}$ specifies that output values 
must be decided among 
at most $k$ different values. 
Finally, the formula $\KNOW$ indicates that, given two facets 
$\sigma$ and $\sigma'$ such that $\sigma \relD{}{A} \sigma'$, 
any process $a\in A$ decides the same output value $d$ at both 
$\sigma$ and $\sigma'$. 


\section{Proving Unsolvability of \texorpdfstring{$k$-Set}{k-Set} Agreement by Logical Method}
\label{sec:IISunsolvable}



Let us revisit the the original proof of Sperner's lemma \cite{Longueville13}, 
which has been applied to show the unsolvability of $k$-set agreement task.
Suppose we have a fully-colored $d+1$-dimensional simplex $X$ 
such that $\coloring(X)=[0,d]$ and also a subdivision of $X$ with Sperner coloring. 
The proof proceeds by induction on dimension.
Assume that, for the fully-colored $d$-dimensional face $Y$ of $X$
such that $\coloring(Y)=[0,d-1]$, 
the subdivision of $Y$ contains an odd number of fully-colored 
$d$-dimensional simplexes. 
The subdivision of $X$ is modeled by a graph, where
the set of graph nodes comprises of the $d+1$-dimensional simplexes of 
the subdivision of $X$
and a special extra node. The extra node is connected by a graph edge
to each $d+1$-dimensional simplex of the subdivision of $X$
that has an $d$-dimensional simplex of the subdivision of $Y$
as a face. Further, a pair of $d+1$-dimensional simplexes are connected by an edge
if and only if they share a common $d$-dimensional face whose
vertexes are distinctly colored by $[0,d-1]$. 
In this graph, a graph node has an odd degree,
if and only if either it is the special extra node or a fully-colored $d+1$-dimensional simplex. 
Therefore, by a simple fact from graph theory that every graph has an even number of nodes of odd degree, 
we conclude that the subdivision of $X$ contains an odd number of fully-colored $d+1$-dimensional simplexes.

In the proceeding proof, instead of counting the number of relevant simplexes in a graph, 
we exploit graph connectivity by using fixpoints and
distributed knowledge modality of epistemic $\mu$-calculus. 
Specifically, assuming the solvability of $k$-set agreement, 
we provide a concrete epistemic $\mu$-calculus formula
that entails the existence of 
a cycle-free, ever-lasting path, which contradicts to the finiteness of the graph. 


\subsection{Proving unsolvability with
    an epistemic \texorpdfstring{$\mu$-calculus}{mu-calculus} formula}
\label{subsec:logicalObst}

To show the unsolvability of the $k$-set agreement task, 
let us assume, by contradiction, that there exists a morphism~$\delta: \IIS{m} \to \ISA{k}$.

For any nonempty subset $A$ of $\colorSet$,  let us define 
$\DEC{A} = 
\bigwedge_{d=0}^{\abs{A}-1}\bigvee_{a\in A}\Pdecide{a}{d}$.
It is easy to see that 
the following epistemic $\mu$-calculus formula
\begin{align} \label{eq:logicalobstruction}
    \Phi_k = {} &
    \nu Z.
    \Bigl[
        \OFUN \wedge \VALID \wedge  
    \bigwedge_{\emptyset\subsetneq A\subseteq \colorSet}
    \bigl(
        \DEC{A} \Rightarrow \ModD{A} (\KNOW \wedge \AGREE{k} \wedge Z)
    \bigr)\Bigr]
\end{align}%
is valid in 
the product update model with factual change $\PUfc{\ISA{k}}$, because
$\OFUN$, $\VALID$, $\KNOW$, $\AGREE{k}$ are all valid formulas in $\PUfc{\ISA{k}}$. 
Then Theorem~\ref{th:PUlogical} implies that $\Phi_k$ is also valid 
in $\fchange{\IIS{m}}{\delta}$. 
In what follows, 
we will show that this cannot happen, that is, 
$\fchange{\IIS{m}}{\delta}, \sigma \not\models \Phi_k$ for some $\sigma\in\Facet(\IIS{m})$. 

\begin{figure}[t] \centering
    \includegraphics[width=1.0\textwidth]{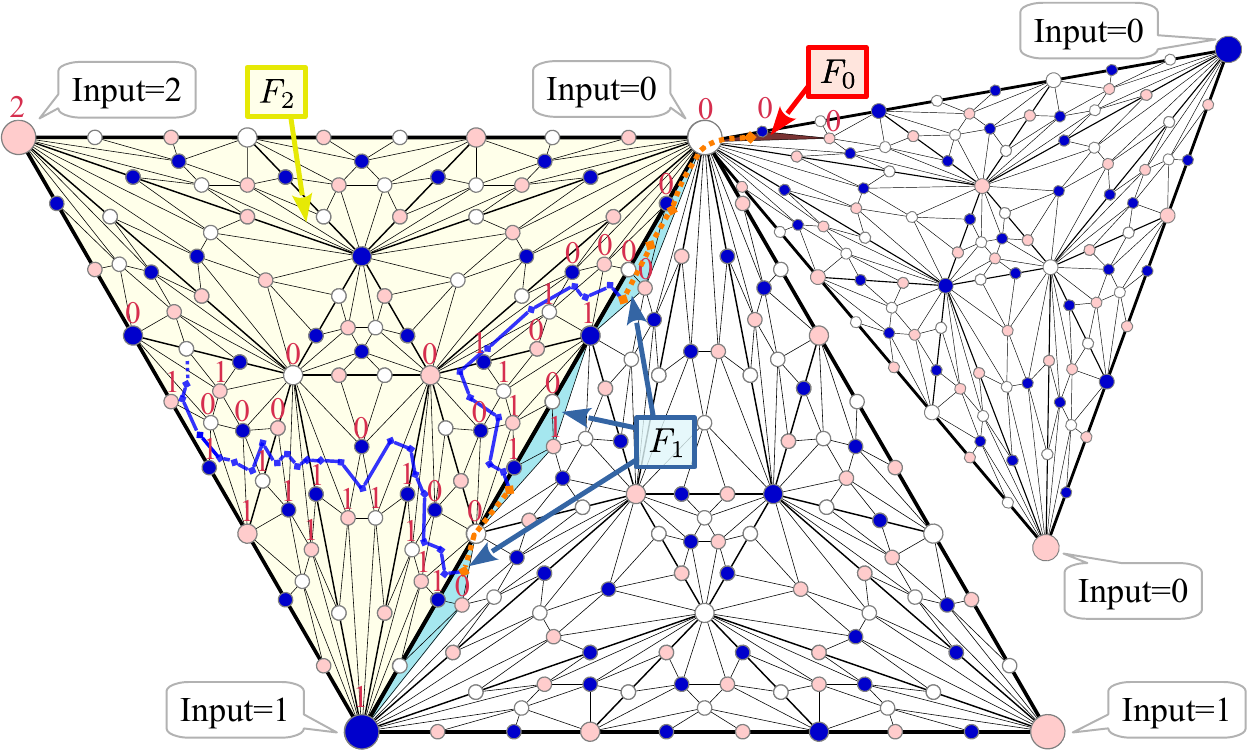}
    \caption{The collections of facets $F_0$, $F_1$, $F_2$ in  
        the product update model $\IISfc{2}{\delta}$. 
        The dotted orange lines and solid blue lines 
        indicate the (incomplete) path of facets connected by the relation 
        $\Rconn{\{0\}}$ and $\Rconn{\{0,1\}}$, respectively.
        The number put on top of each vertex is the output value.}
        \label{fig:IIS2fc}
\end{figure}


Let 
$I_0, \ldots, I_n \in \Facet(\cplI)$ be facets of the input complex $\cplI$ defined by
\begin{align*}
    I_d ={} & \{(i,i) \mid i \in [0,d] \} \cup \{(i,d) \mid i\in [d+1, n] \} 
\end{align*}
and also $F_0, \ldots, F_{n} \subseteq \Facet(\IISfc{m}{\delta})$ 
be collections of the facets in $\IISfc{m}{\delta}$ defined by
\begin{align*}
    \quad F_d ={} & \{
        I_d \fX \gamma_1 \fX \cdots \fX \gamma_m \mid
    \parbox[t]{.75\textwidth}{%
        $\gamma_i = \osp{A_{i,1} {\mid} \cdots {\mid} A_{i,r_i} {\mid} d+1 {\mid} \cdots {\mid} n}$ ($r_i\geq 1$)
        is an ordered set partition of $[0,n]\}$.}    
\end{align*}

Each collection $F_d$ is intended to 
represent the graph nodes in the proof of Sperner's lemma.
Each facet in $F_d$ 
has a $d$-dimensional face that is a simplex
of the subdivision of the $d$-dimensional face $I_d'=\{ (i,i) \mid i \in [0,d]\}$ of $I_d$. 
Conversely, each $d$-dimensional simplex of the subdivision of $I_d'$
is a face of a unique facet $I_d \fX \gamma_1 \fX \cdots \fX \gamma_m\in F_d$,
where the uniqueness is due to
the complementary vertexes $\{(i,d)\mid i\in [d+1,n]\}$ in $I_d$ and 
the trailing sequence ${\mid} d+1 {\mid} \cdots {\mid} n\rangle$ 
in each ordered set partition $\gamma_j$.

We define an irreflexive symmetric relation 
$\Rconn{A}$ over $\bigcup_{d=0}^k F_d$, which corresponds to  
the graph edges defined in the proof of Sperner's lemma: 
For every $\sigma,\sigma'\in \bigcup_{d=0}^k F_d$,  
$\sigma \Rconn{A} \sigma'$ holds iff 
\begin{itemize}
    \item $\sigma\relD{}{A}\sigma'$ and $\sigma\neq \sigma'$,
    \item $A\subseteq [\,0,\abs{A}\,]$, 
    \item $\sigma \in F_d$, $\sigma' \in F_e$ for some $d, e$ such that
        $\max(d,e)=\abs{A}$ and $\abs{d-e}\leq 1$, and 
    \item $\IISfc{m}{\delta}, \sigma\models \DEC{A}$ 
        and $\IISfc{m}{\delta}, \sigma'\models \DEC{A}$.
\end{itemize}

Figure~\ref{fig:IIS2fc} depicts, for the case of three processes, the collections of facets  
$F_0$, $F_1$, $F_2$ in $\IIS{2}$ and how the facets are related by $\Rconn{A}$. 
Note that, unlike the proof of Sperner's lemma where a graph is constructed
for each different dimension,
we define the relation $\Rconn{A}$ for all dimensions at once. 
Instead of putting a special graph node for each dimension, we allow 
a pair of facets $\sigma\in F_d$ and $\sigma\in F_{d+1}'$ to be related,
if they share a $d$-dimensional common face whose vertexes are colored by $[0,d-1]$. 

Particularly when $d=1$, we see that 
$F_0$ consists of a single facet $\sigma_0$ that
is related solely with a single facet in~$F_1$.
\begin{lemma} \label{lemma:Fzero}
    Let $\gamma_1=\cdots=\gamma_m = \osp{0{\mid}1{\mid}\cdots{\mid}n}$. 
    We have 
    $F_0 = \{ \sigma_0 \}$, where $\sigma_0 = I_0\fX \gamma_1 \fX\cdots\fX \gamma_m$.
    Also, $\sigma_0 \Rconn{A} \sigma'$ iff 
    $A=\{0\}$ and $\sigma' = I_1 \fX\gamma_1 \fX\cdots\fX \gamma_m$.
\end{lemma}



More generally, up to certain condition that is compelled by the formula $\Phi_k$,
each facet in $F_0,F_1,\ldots,F_k$ other than $\sigma_0$ is 
related with zero or two facets by $\Rconn{A}$.
\begin{proposition} \label{prop:bidegree}    
    Suppose $\IISfc{m}{\delta}\models \VALID$. 
    For any $\sigma \in F_d$ ($1\leq d\leq k$) such that
    $\IISfc{m}{\delta},\sigma \models \OFUN\wedge\VALID\wedge\AGREE{k}\wedge\KNOW$, 
    the number of facets $\sigma' \in \bigcup_{d=0}^k F_d$ satisfying $\sigma \Rconn{A} \sigma'$
    for some $A$ is either 0 or~2. 
\end{proposition}

\begin{proof}
    Given $\sigma \in F_d$ ($0<d\leq k$) such that 
    $\fchange{\IIS{m}}{\delta}, \sigma \models \OFUN\wedge\VALID\wedge\AGREE{k}\wedge\KNOW$, 
    let us define $f_{\sigma}(a)=b$ iff $\fchange{\IIS{m}}{\delta},\sigma\models \Pdecide{a}{b}$. 
    Then $f_{\sigma}$ is a well-defined total function from $[0,n]$ to $[0,d]$, because  
    $\fchange{\IIS{m}}{\delta},\sigma \models \OFUN\wedge \VALID$.
    
    Consider the case $[0,d-1]\not\subseteq f_{\sigma}([0,d])$. 
    Assume, by contradiction, $\sigma \Rconn{A}\tau$ holds for some $A$ and
    $\tau\in F_e$ ($0< e\leq k$).  
    Then $[0,d-1]\subseteq f_{\sigma}(A)$ 
    holds from $\fchange{\IIS{m}}{\delta},\sigma \models \DEC{A}$.
    If either $e =d-1$ or $e=d$, $\sigma \Rconn{A}\tau$ implies $A\subseteq [0,d]$;
    If $e=d+1$, 
    $\sigma \relD{}{A} \tau$ implies $A\subseteq [0,d]$, 
    because $\sigma$ and $\tau$ are a subdivision of $I_d$ and $I_{d+1}$, respectively, 
    and $\coloring(I_d \cap I_{d+1})=[0,d]$. 
    In either case, we have $[0,d-1] \subseteq f_{\sigma}(A)\subseteq f_{\sigma}([0,d])$,
    a contradiction. 
    Therefore $\sigma \Rconn{A} \tau$ does not hold for any $\tau$ and $A$.

    Consider the other case $[0,d-1]\subseteq f_{\sigma}([0,d])$.
    Then, either of the following hold: 
    \begin{enumerate}[(a)] 
        \item  \label{case:full} $d\in f_{\sigma}([0,d])$. 
        In this case, $f_{\sigma}$ is a bijection on $[0,d]$. 
        In order for $\tau \Rconn{A} \sigma$ to hold for some $\sigma$, 
        $A$ must be either $B_0=[0,d]$ or $B_1$ such that 
        such that $B_1\subseteq [0,d]$, $\abs{B_1}=d$, and 
        $f_{\sigma}(B_1)=[0,d-1]$.
        (This implies $d<k$, because $\fchange{\IIS{m}}{\delta},\sigma\models \AGREE{k}$.)

        \item \label{case:twofaces} $d\not\in f_{\sigma}([0,d])$. 
        In this case, $f_{\sigma}$ is a surjection from $[0,d]$ to $[0,d-1]$. 
        In order for $\tau \Rconn{A} \sigma$ to hold for some $\sigma$, 
        $A$ must be either of $B_2, B_3 \subseteq [0,d]$ such that
        $\abs{B_2}=\abs{B_3}=d$, 
        $\abs{B_2\cap B_3}=d-1$,  and $f_{\sigma}(B_2)=f_{\sigma}(B_3)=[0,d-1]$. 
    \end{enumerate}

    Suppose 
    $\sigma = I_d \fX \gamma_1 \fX  \cdots \fX$ $\gamma_m$ is a facet in $F_d$, 
    where $\gamma_j = \osp{A_{j,1}{\mid}\cdots {\mid} A_{j,r_j} {\mid} d{+}1 {\mid}\cdots {\mid} n}$
    for each $j$. 
    Let us show that $\sigma$ has a unique facet related by $\Rconn{B_i}$, for each $i\in \{0,1\}$
    in the case of~(\ref{case:full}) and also for each 
    $i\in \{2,3\}$ in the case of~(\ref{case:twofaces}).
    For the case of $B_0$,
    $\tau = I_{e}\fX \gamma_1 \fX  \cdots \fX \gamma_m \in F_e$ is the sole facet satisfying 
    $\sigma \relD{}{B_0} \tau$, because $\sigma\relD{}{B_0}\tau'$ does not hold for any
    facet $\tau'\in F_d$ by Proposition~\ref{prop:opsincident}. 
    For the remining cases,  we show 
    $B\subseteq [0,d]$, $\abs{B}=d$, and $f_{\sigma}(B)=[0,d-1]$ holds, where 
    $B$ is either $B_1, B_2$, or $B_3$.     
    By Proposition~\ref{prop:opsincident}, we have three cases to consider.
    \begin{itemize}
        \item $B=[0,d-1]$ and $A_{j,r_j}= \{d\}$ for each $j$. \\
        There is no facet $\tau\in F_{d}$ satisfying $\sigma \relD{}{B} \tau$
        by Proposition~\ref{prop:opsincident}. Hence
        the sole facet satisfying $\sigma \relD{}{B} \tau$ is 
        $\tau = I_{e}\fX \gamma_1 \fX  \cdots \fX \gamma_m \in F_{e}$.
        \item $B=[0,d]\setminus \{c\}$ for some $c\in [0,d-1]$ and $A_{j,r_j}= \{c\}$ for each $j$. \\
        This case cannot happen. Suppose there exists such $c$. 
        Define $J= \{(i,v) \mid (i,v)\in I_d, i\neq c\}\cup \{ (c,d)\}$. 
        Then $I_d \relD{}{B} J$ and hence $\sigma \relD{}{B} \tau$, 
        where $\tau = J\fX \gamma_1 \fX  \cdots \fX \gamma_m$. 
        Since $f_{\sigma}(B)=[0,d-1]$, it holds that 
        $\fchange{\IIS{m}}{\delta},\sigma\models \Pdecide{a}{c}$ for some $a\in B$.
        Then, $\fchange{\IIS{m}}{\delta},\sigma\models\KNOW$ implies 
        $\fchange{\IIS{m}}{\delta},\tau\models \Pdecide{a}{c}$. 
        However, since $\fchange{\IIS{m}}{\delta},\tau\models\VALID$, we have
        $\fchange{\IIS{m}}{\delta},\tau\models\Pinput{b}{c}$ for some $b$, 
        which contradicts to the definition of $J$. 
        \item There exists $j$ such that either $\abs{A_{j,r_j}}>1$
        or $B=[0,d]\setminus A_{j,t}$ where $A_{j,t}= \{c\}$ for some $t$ ($1\leq t< r_j$) and 
        $c\in [0,d]$. \\
        Let $\ell$ be the minimum of such $j$. 
        Then, by Proposition~\ref{prop:opsincident}, the sole facet satisfying $\sigma \relD{}{B} \tau$
        is $\tau = I_{d}\fX \gamma_1' \fX  \cdots \fX \gamma_m' \in F_{d}$,
        where $\gamma_\ell' = \flip{B}{\gamma_\ell}$ and $\gamma_i'=\gamma_i$ for every 
        $i$ ($i\neq \ell$).         
    \end{itemize}

    To see $\sigma \Rconn{B_i} \tau$ holds for each $i\in \{0,1,2,3\}$, it remains to show
    $\fchange{\IIS{m}}{\delta},\sigma \models \DEC{B_i}$  and
    $\fchange{\IIS{m}}{\delta},\tau \models \DEC{B_i}$. 
    The former immediately follows from the definition of $f_{\sigma}$; 
    The latter follows from the former and 
    $\fchange{\IIS{m}}{\delta},\sigma\models \KNOW$. 
\end{proof}

\begin{theorem} \label{th:IISm}
    Let $\ISA{k}$ and $\IIS{m}$ be the product update models for $k$-set agreement task    
    and $m$-iterated immediate snapshot protocol, respectively.
    If $k\leq n$, there exists no morphism $\delta$ from $\IIS{m}$ to $\ISA{k}$, meaning that
    $k$-set agreement task is not solvable by the iterated immediate snapshot protocol. 
\end{theorem}
\begin{proof}
    Suppose, by contradiction, there exists such a morphism $\delta$.
    By Theorem~\ref{th:PUlogical}, we have 
    $\PUfc{\ISA{k}}\models \Phi_k$, and hence $\IISfc{m}{\delta} \models \Phi_k$.    
    In particular, $\IISfc{m}{\delta} \models \VALID$.


    Hence it must hold that $\IISfc{m}{\delta}, \sigma_0 \models \Phi_k$,
    where $\sigma_0$ is the sole facet of $F_0$. 
    Unfolding $\nu Z$ once, we have
    $\IISfc{m}{\delta}, \sigma_0 \models \OFUN \wedge \VALID$, which 
    implies $\IISfc{m}{\delta}, \sigma_0 \models \Pdecide{0}{0}$ and hence
    $\IISfc{m}{\delta}, \sigma_0 \models \DEC{\{0\}}$.
    Then, by Lemma~\ref{lemma:Fzero}, we can see
    there exists $\sigma_1\in F_1$ such that $\sigma_0 \Rconn{\{0\}}\sigma_1$
    and hence $\IISfc{m}{\delta}, \sigma_1 \models 
    \OFUN \wedge \VALID \wedge \KNOW \wedge \AGREE{k} \wedge \Phi_k$.

    Let us show that, for every $\ell\geq 1$, there exists a path of facets 
    $\sigma_0 \Rconn{A_0} \sigma_1 \Rconn{A_1} \cdots 
    \Rconn{A_{\ell-1}} \sigma_\ell$ such that 
    $\sigma_0$, $\sigma_1, \ldots, \sigma_\ell$ are pairwisely distinct facets
    in $\bigcup_{i=0}^k F_k$ 
    and also $\IISfc{m}{\delta}, \sigma_i \models 
    \OFUN \wedge \VALID \wedge \KNOW \wedge \AGREE{k} \wedge \Phi_k$ holds
    for every $\sigma_i$.
    We show this by induction on $\ell$. 
    We have already examined the case $l=1$. Suppose the claim holds up to~$\ell$,
    that is, $\sigma_{\ell-1} \Rconn{A_{\ell-1}} \sigma_l$ and $\IISfc{m}{\delta}, \sigma_\ell \models 
    \OFUN \wedge \VALID \wedge \KNOW \wedge \AGREE{k} \wedge \Phi_k$. 
    By Proposition~\ref{prop:bidegree}, $\sigma_\ell$ must have 
    a facet $\sigma_{\ell+1}$, other than $\sigma_{\ell-1}$,  such that
    $\sigma_{\ell} \Rconn{A_{\ell}} \sigma_{\ell+1}$ for an appropriate $A_\ell$.
    The relation $\sigma_{\ell} \Rconn{A_{\ell}} \sigma_{\ell+1}$ implies 
    $\IISfc{m}{\delta}, \sigma_\ell \models \DEC{A_\ell}$
    and hence we have $\IISfc{m}{\delta}, \sigma_{\ell+1}\models
    \OFUN \wedge \VALID \wedge \KNOW \wedge \AGREE{k} \wedge \Phi_k$ again 
    for $\sigma_{\ell+1}$, by unfolding~$\nu Z$.

    Finally, let us show that $\sigma_{\ell+1}$ is distinct from any other preceding facets. 
    If $\ell=1$, this immediately follows from Proposition~\ref{prop:bidegree}.
    Suppose $\ell>1$ and $\sigma_{\ell+1} = \sigma_{j}$ for some $j$ ($1\leq j<\ell-1$). 
    Then, by $\sigma_{j-1} \Rconn{A_{j}} \sigma_j \Rconn{A_{j+1}} \sigma_{j+1}$
    and Proposition~\ref{prop:bidegree},  
    either $\sigma_{j-1}$ or $\sigma_{j+1}$ must have been $\sigma_{\ell}$, which contradicts 
    the induction hypothesis. 

    The existence of path of distinct facets of arbitrary length
    contradicts to the fact that $\bigcup_{i=0}^k F_k$ is a finite set.
    Therefore there exists no morphism $\delta$. 
\end{proof}


\section{Set Agreement in Submodel: \texorpdfstring{$k$-Concurrency}{k-Concurrency}}
\label{sec:kconc}

\begin{figure}
    \centering
    \includegraphics[width=0.33\textwidth]{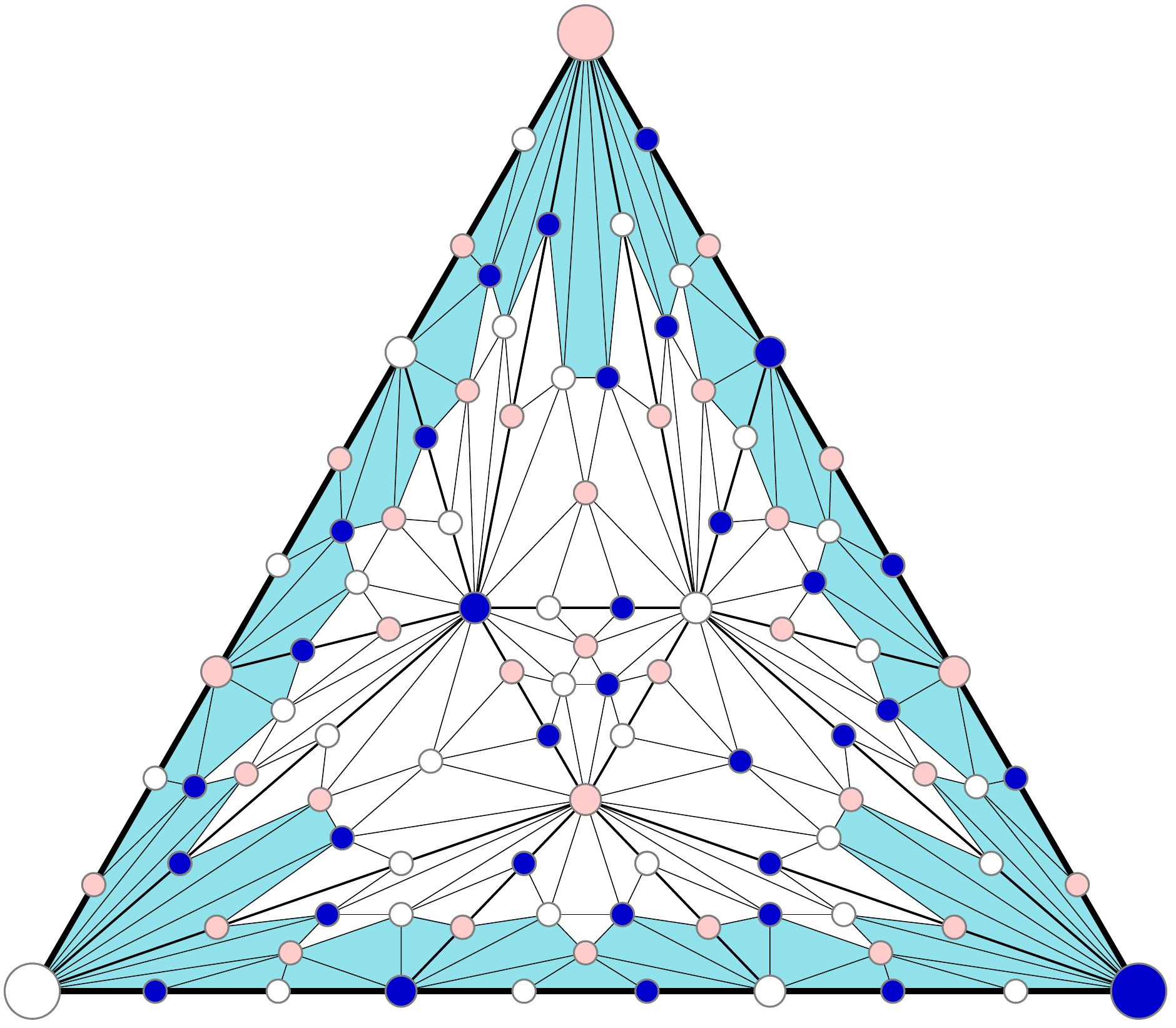}
    \caption{The 2-concurrency model \texorpdfstring{$\kconc{2}$}{R2} for a 3-process system, 
    a submodel of $\IIS{2}$ restricted to the shaded facets}
    \label{fig:kconc}
\end{figure}

The unsolvability argument carried out in Section~\ref{sec:IISunsolvable}
can be applied to a submodel of $\IIS{m}$, 
if the submodel keeps a relevant set of facets in $\IIS{m}$. 
In this section, let us examine $k$-concurrency model \cite{GafniHeKuznetsovRieutord16}, 
as an instance of such a submodel. 

The \keywd{$k$-concurrency model}, denoted by $\kconc{k}$, is a submodel of $\IIS{2}$, 
where the facets in $\IIS{2}$ are confined to those 
corresponding to an execution in which 
at most $k$ out of $n+1$ concurrent processes are allowed to execute
the two consecutive iterations of immediate snapshot protocol simultaneously.

Following \cite{GafniHeKuznetsovRieutord16}, 
we formally define the $k$-concurrency submodel $\kconc{k}$ as follows.
For an ordered set partition $\gamma=
\osp{A_1 {\mid} A_2{\mid} \cdots{\mid} A_r}$
of $\colorSet$ and $a\in \colorSet$, 
we define $\viewOf{a}(\gamma)= \bigcup_{i=1}^{t} A_i$, where $a\in A_t$.  
Also, for a facet $X\fX\gamma_1\fX\gamma_2 \in \Facet(\IIS{2})$
and a vertex of color $a\in \colorSet$, 
we define the \keywd{carrier set} for $a$
by $\carrier{a}(X\fX\gamma_1\fX\gamma_2) =\bigcup_{b\in\viewOf{a}(\gamma_2)}\viewOf{b}(\gamma_1)$.
The contention sets \cite{GafniHeKuznetsovRieutord16} for a facet $\sigma\in \IIS{2}$ are defined by
$\Contention(\sigma) = \{ A\subseteq \colorSet \mid \carrier{a}(\sigma)
=\bigcup_{b\in A} \carrier{b}(\sigma)$ for every $a\in A\}$. 
Then the $k$-concurrency model $\kconc{k}$ is a submodel of $\IIS{2}$ with
a restricted set of facets
$\Facet(\kconc{k}) = 
    \bigl\{ \sigma\in \IIS{2} \bigm\vert \abs{A}\leq k \text{ for every $A\in\Contention(\sigma)$} \bigr\}$.

Figure~\ref{fig:kconc} illustrates 
the $2$-concurrency submodel $\kconc{2}$ for a 3-process system. 
The submodel contains all the facets relevant for the unsolvability argument 
for $1$-set consensus, namely, $F_0$ and $F_1$ in Section~\ref{sec:IISunsolvable}. 
We can formally check this containment property as follows.
Remember that all the facets considered 
in the previous section for $k$-set agreement (with $m=2$) has 
the form $I_d\fX \gamma_1\fX \gamma_2$ ($d\leq k$) 
where both of the ordered set partitions $\gamma_1, \gamma_2$ are of the form
$\osp{A_1 {\mid} A_2{\mid} \cdots{\mid} A_r {\mid} \,d+1\,  {\mid}\cdots{\mid} n}$. 
For such a facet $\sigma$, $\carrier{a}(\sigma)= [0,a]$ if $a\in[d+1,n]$,
while $\carrier{a}(\sigma)\subseteq [0,d]$ if $a\in[0,d]$.
Therefore 
$\abs{A}\leq d+1\leq k+1$ for every $A\in\Contention(\sigma)$
and hence $I_d\fX \gamma_1\fX \gamma_2\in \Facet(\kconc{k+1})$.

By the same argument in the preceding section, 
we can show that $k-1$ (or fewer) set agreement task is not solvable
in the $k$-concurrency model.
\begin{theorem} 
    Let $\ISA{\ell}$ be the product update model for $\ell$-set agreement task,
    $\kconc{k}$ be the $k$-concurrency model.
    Then there is no morphism $\delta$ from $\kconc{k}$ to $\ISA{\ell}$, 
    if $\ell <k$. 
\end{theorem}


\bibliographystyle{plain}

\bibliography{distrib}

\end{document}